\documentclass{article}

\usepackage[centertags]{amsmath}
\usepackage{amsfonts,amsthm,amssymb}
\usepackage{amssymb}
\usepackage{amsmath}
\usepackage{graphicx}
\usepackage{booktabs}
\usepackage{appendix}
\usepackage[hmargin=2.6cm,vmargin=2.6cm]{geometry}
\usepackage{etoolbox}
\usepackage{tikz}
\usepackage{soul}
\usepackage{comment}
\usepackage{float}
\usepackage{setspace}
\usepackage{xcolor,framed}
\usepackage{soul}   
\setstcolor{red}
\usepackage[round]{natbib}
\usepackage{hyperref}
\hypersetup{colorlinks,citecolor=blue}
\usepackage{threeparttable}
\usepackage{array}
\usepackage{booktabs}
\usepackage{blindtext}
\usepackage{amsmath}
\usepackage{amsthm}
\usepackage{amssymb}
\usepackage{mathtools}
\usepackage{bbold}
\usepackage{algorithm2e}
\usepackage{graphicx}
\usepackage{subcaption}
\usepackage{pgfplots}
\pgfplotsset{compat=newest} 

\DeclareMathOperator*{\E}{\mathbb{E}}
\DeclareMathOperator*{\R}{\mathbb{R}}
\DeclareMathOperator*{\prob}{\text{I\kern-0.15em P}}
\DeclarePairedDelimiterX\set[1]\lbrace\rbrace{#1}

\DeclareMathOperator*{\argmax}{argmax}

\DeclareMathOperator*{\supp}{\operatorname{supp}}
\DeclareMathOperator*{\us}{s}

\newtheorem{theorem}{Theorem}
\newtheorem{definition}{Definition}
\newtheorem{lemma}{Lemma}
\newtheorem{proposition}{Proposition}

\newtheorem{corollary}{Corollary}
\newtheorem{example}{Example}
\newtheorem{remark}{Remark}

\title{Robust Price Discrimination}

\author{
    Itai Arieli\thanks{%
    University of Toronto and Technion - Israel Institute of Technology,
    iarieli@technion.ac.il},
    Yakov Babichenko\thanks{%
    Technion - Israel Institute of Technology,
    yakovbab@technion.ac.il},
    Omer Madmon\thanks{%
    Technion - Israel Institute of Technology,
    omermadmon@campus.technion.ac.il},
    Moshe Tennenholtz\thanks{%
    Technion - Israel Institute of Technology,
    moshet@technion.ac.il}
}

\begin{document}
\maketitle

\begin{abstract}
We consider a model of third-degree price discrimination where the seller's product valuation is unknown to the market designer, who aims to maximize buyer surplus by revealing buyer valuation information. Our main result shows that the regret is bounded by a $\frac{1}{e}$-fraction of the optimal buyer surplus when the seller has zero valuation for the product. This bound is attained by randomly drawing a seller valuation and applying the segmentation of Bergemann et al. (2015) with respect to the drawn valuation. We show that this bound is tight in the case of binary buyer valuation.
\end{abstract}

\section{Introduction}
\label{intro_section}
The celebrated paper of \cite{bergemann2015} considers a setting in which a product is sold via the posted price mechanism. The interaction involves three agents. A \emph{buyer} whose value for the product is drawn according to a commonly known distribution $b\sim \mu \in \Delta(\mathbb{R}_+)$, supported on $n$ buyer types. A \emph{seller} whose value for the product is $\us$ (e.g., the seller has an outside option of selling the product for a price of $\us$). The third entity is a \emph{market designer} (a \emph{designer} for short) who knows the buyer's value for the product $b$ and is allowed to credibly reveal information about $b$ to the seller. 

\cite{bergemann2015} characterize the possible (buyer, seller) utility profiles that may arise in the above interaction under some revelation policy of the designer (i.e., under some \emph{market segmentation}). Arguably, the most interesting case arises when the designer's objectives are aligned with those of the buyer; namely when the designer tries to maximize the (ex-ante) buyer's surplus. Surprisingly, \cite{bergemann2015} shows that a careful choice of market segmentation might yield the entire surplus to the buyer subject to the obvious constraint that the seller must get at least her no-information surplus (because the seller can ignore the designer's information).

This very elegant result relies, however, on the assumption that the designer knows the seller's value for the product $\us$ (in such a case, $\us$ can be normalized to $\us=0$). 
In many natural scenarios, the knowledge of $\us$ (e.g., the outside options of the seller) by the designer might be a demanding requirement. 
To motivate the need for relaxing this assumption, one might think of the designer entity as a purchasing group or a buyers club. The purchasing group perfectly knows the valuations of its members. However, once a seller arrives, the group might have uncertainties about the seller's outside option. Moreover, notice that in this example, the objective is indeed buyer surplus maximization.

In this paper, our goal is to explore this interaction in the case where $\us$ is unknown to the designer. While \cite{bergemann2015} is an information design problem in a setting of trade with incomplete information on one side (i.e., only buyer valuation is private information) our setting can be viewed as a first step towards understanding information design in the classical bilateral trade setting (see \citealp{myerson1983efficient}) in which both buyer and seller valuations are private information.\footnote{A recent paper by \cite{shota2023buyer} also adopts an information design approach to bilateral trade. In \cite{shota2023buyer}, the buyer does not know her value for the product but, instead, receives a recommendation on whether to adopt the product from a designer.}

A natural approach in this setting is the robust perspective, where the designer aims to come up with a market segmentation that performs well \emph{for all} $\us$. The robust approach has been adopted to various economic setups, such as auctions (\citealp{bose2006optimal}), monopoly pricing (\citealp{bergemann2011robust,wong2020optimal}), principal-agent (\citealp{carroll2015robustness,walton2022general}) and persuasion (\citealp{dworczak2022preparing,babichenko2022regret}).

There are two leading branches of the robust paradigm. The first one is the max-min paradigm; the market segmentation should yield a high utility for the buyer for all $\us$. This approach is not very insightful. An adversarial choice of $\us$ that is above all elements in $\supp (\mu)$ leads to no trade, and hence, to a zero surplus for the buyer. The second robustness paradigm is regret; the market segmentation should yield a utility as close as possible to the hypothetical scenario in which $\us$ was known to the designer.\footnote{Previous work demonstrates how the robust approach
can be used to provide insights in realistic economic settings. \cite{guo2023regret} used the regret-minimization framework to provide an intuition for why
agents should be incentivized to propose multiple projects in a setup in which a principal-agent project selection setup. 
\cite{guo2025robust} shows that in the context of monopoly regulation, regret-minimizing policies can be simpler and more transparent than those relying on policymaker knowledge.
\cite{arieli2025random} justifies the random dictator mechanism for forecast aggregation, by showing that the mechanism is universally optimal across various robustness notions. These results motivate applying regret-minimization to price discrimination as a path to robust and interpretable segmentation.}

From a technical point of view, finding a regret-minimizing segmentation is equivalent to solving a zero-sum game between the designer (who chooses a segmentation) and an adversary (who chooses $\us$). This game is difficult to solve analytically since the strategy space of the designer is of infinite dimension.\footnote{Notice that, unlike in the standard model of \cite{bergemann2015}, the number of segments on which the optimal segmentation is supported can be unbounded. In fact, we will show that the regret-minimizing segmentation indeed uses infinitely many signals for binary valuations.}
Nevertheless, we are able to provide valuable insights--in particular, we introduce a simple, intuitive, and tractable approach to constructing a robust market segmentation with non-trivial robustness guarantees.

Our main result (Theorem \ref{main_res_theorem}) shows that for every $\mu$ there exists a market segmentation that ensures regret of $U^*(0)/e$, where $U^*(0)$ is the buyer's surplus in the case where $\us=0$ and is known to the designer. That is, the difference between the performance of our ignorant designer and that of the hypothetical designer who knows $\us$ will never exceed $U^*(0)/e$. This market segmentation is the regret-minimizing one whenever the buyer type is binary, i.e., $n=2$ (Theorem \ref{main_res_bintype_theorem}). The market segmentation that achieves the regret bound is quite intuitive: The designer does not know $\us$. A simple idea is to "guess $\us$" according to a carefully chosen distribution $s_D \sim g \in \Delta(\mathbb{R}_+)$, and thereafter to create a segmentation that is optimal for the case $\us=s_D$ exactly as described in \cite{bergemann2015}. This two-step procedure creates yet another market segmentation and we call the class of these segmentations \emph{BBM segmentations}. Interestingly, it turns out that the density function of the optimal BBM strategy can be expressed as the \emph{hazard function} of the survival function corresponding to the optimal buyer surplus function, providing an intuitive interpretation of the robust segmentation strategy.

Besides the theoretical results, we also present empirical evidence for the effectiveness of our BBM segmentation. Empirical studies demonstrate that realistic valuations fit best Lognormal distributions or Pareto distributions (see \citealp{coad2009distribution}).
We show that in an experimental setting in which buyer and seller valuations are drawn independently from a shared distribution (which can be either Lognormal or Pareto), the actual expected difference between the optimal surplus and the surplus achieved by our segmentation is even lower than the theoretical bound we prove in Theorem \ref{main_res_theorem}.

\paragraph{Techniques} The idea behind the proof of the main results (Theorems \ref{main_res_theorem} and \ref{main_res_bintype_theorem}) is as follows. 
While solving the zero-sum game between the designer and the adversary is involved (due to the high-dimensionality of the designer's strategy space), it can be simplified by restricting the designer to BBM segmentations. In this case, the strategy space of the designer turns to a single-dimensional one (just the choice of $s_D \in \mathbb{R}_{+}$). It turns out that the utilities in this zero-sum game (as a function of $\us$ and $s_D$) have relatively clean expressions, but are not clean enough to be able to be solved explicitly. We bound from above the utilities in the zero-sum game by even cleaner expressions. The latter formula is so clean that we are able to perform the entire equilibrium analysis of this game (see Lemma \ref{main_tech_lemma}). These arguments are sufficient to deduce Theorem \ref{main_res_theorem}. 

For the proof of Theorem \ref{main_res_bintype_theorem} we observe that along the proof of Theorem \ref{main_res_theorem} we have made two relaxations. First, we restricted the designer to BBM segmentations. Second, we have bounded from above the utilities of the actual zero-sum game. It turns out that the second relaxation is not actually a relaxation in the case of $n=2$. It is the actual zero-sum game. Therefore, a natural attempt would be to focus on the mixed strategy of the adversary in this zero-sum game. Now we allow the designer to use all the segmentations (not only the BBM ones) and we show that even if her strategy space is richer (i.e., not only BBM segmentations) she still cannot gain more than the value of the zero-sum game. Such an analysis is tractable because in the case of $n=2$, the designer's best reply problem boils down to a concavification of a single-dimensional function. This somewhat surprising observation is sufficient to deduce Theorem \ref{main_res_bintype_theorem}. The observation is somewhat surprising because there exist mixed strategies of the adversary for which all best replies of the designer do not belong to the BBM segmentations class (see Example \ref{example2}). The specific mixed strategy that is optimal for the zero-sum game turns out to have BBM segmentation best-reply.
Since establishing the regret lower bound for $n=2$ relies on the concavification approach, we tend to believe that providing such a lower bound for the general case is difficult, and requires a different set of techniques. We therefore leave this interesting open problem outside the scope of this paper.

\paragraph{Paper structure} 
Section \ref{model_section} introduces the price discrimination model of \cite{bergemann2015}, which assumes a known seller valuation. 
We then introduce and analyze our incomplete information model in Section \ref{robust_subsection}. We state and prove our main result, which is an upper bound on the overall regret using our BBM segmentation (Theorem \ref{main_res_theorem}).
We then show in Section \ref{binary_type_section} that this bound is tight for the binary buyer type case (Theorem \ref{main_res_bintype_theorem}).
Section \ref{expriments_section} introduces our experimental results in which we evaluate the performance of our approach compared to the optimal benchmark for a specific class of distributions that are considered as reflecting the actual distribution of valuations in realistic settings. 
In Section \ref{sec:extensions}, we extend our regret upper bound to two settings:
(i) the case where the designer knows that the seller’s valuation lies within a given interval (Subsection \ref{subsec: restricted_seller_val}), and (ii) the case of weighted generalized regret (Subsection \ref{subsec:generalized_regret}).
We then conclude in Section \ref{conclusions_section}. Proofs of all technical lemmas are deferred to the appendix.

\subsection{Related Work}
\label{related_work_section}

\paragraph{Third-degree price discrimination} A fundamental economic question is how third-degree price discrimination affects consumer surplus, producer surplus, and social welfare (see, e.g., the classic work of \citealp{pigou2017economics}).\footnote{In "third-degree" price discrimination, the market designer divides the market into separate segments, where the seller may charge different prices in each segment. In contrast, "first-degree" price discrimination refers to a situation in which the seller is fully informed of the buyer's value (hence charges this value as the price), and in "second-degree" price discrimination the seller sells goods that are similar, but may vary in quality, at different prices. Throughout the paper, we use the term 'price discrimination' to describe a model of third-degree price discrimination.} Our work extends the work of \cite{bergemann2015}, in which the price discrimination model is studied from a buyer surplus maximization perspective. \cite{bergemann2015} introduced an algorithm for finding a buyer surplus maximizing market segmentation, which can be computed efficiently. 
Several works have then extended the standard model, and provided either exact or approximate buyer-optimal segmentation under different assumptions (\citealp{shen2018closed_pd1}, \citealp{cai2020third_pd2}, \citealp{mao2021interactive_pd3}, \citealp{bergemann2022calibrated_pd4}, \citealp{alijani2022limits_pd5}, \citealp{ko2022optimal_pd6},
\citealp{bergemann2024unified}).

While maximizing surplus often yields an unfair outcome for the buyers, an alternative promising line of research focuses on \emph{fair price discrimination} (\citealp{flammini2021fair_fpd1}, \citealp{cohen2022price_fpd2}). In particular, \cite{banerjee2023fair_fpd3} prove the existence of a segmentation, different than the one of \cite{bergemann2015}, that simultaneously approximates a large set of welfare functions (including utilitarian welfare, Nash welfare and the min-max welfare). Their approach can be viewed as another notion of robustness other than regret-minimization, since the segmentation is robust to the actual welfare function.

Closer to our work, \cite{cummings2020algorithmic} analyzed several variations of the price discrimination model, in which the market designer only has a noisy signal about the \emph{buyer's} valuation. Our work completes the picture by studying the case in which the uncertainty of the designer is with respect to the \emph{seller's} valuation.

\paragraph{Robust Bayesian persuasion} As we discuss in the paper, the model of price discrimination closely relates to the Bayesian persuasion model introduced by \cite{kamenica2011bayesian}. Bayesian persuasion refers to a situation in which an informed sender aims to influence the decision of an uninformed receiver by designing a signaling scheme. One rigid assumption required in the standard model is that the sender knows the receiver's type (i.e., utility function), and uses this information to construct an optimal signaling policy. Several works took different approaches to relax this restricting assumption: \cite{arieli2023reputationbased} took a natural \emph{Bayesian} approach, meaning the receiver is sampled from a commonly known prior distribution; \cite{dworczak2022preparing} took a \emph{minmax} approach, which measures the absolute performance of a signaling scheme that does not rely on knowing the receiver's type; \cite{castiglioni2020online} considered an \emph{online learning} framework, in which the sender repeatedly faces an adversarially-chosen receiver whose type is unknown, and receives either a full-information or partial-information feedback. 

Closest to our work, \cite{babichenko2022regret} studied a \emph{regret-minimizing} Bayesian persuasion model, in which the performance of a signaling policy is determined according to the worst-case difference between the optimal utility the sender could obtain \emph{had she known} the receiver's type, and the actual utility she obtains without having access to this information. 
The class of utilities that are studied in \cite{babichenko2022regret} is different from the one that arises in a price discrimination model. \cite{babichenko2022regret} consider a receiver with binary decisions and state-independent utilities for the sender. In our case, the number of actions for the seller (receiver) is $n$. Moreover, the sender's utility is not state-independent. Even in the binary valuation case (i.e., $n=2$) the class of the sender's indirect utilities in our case might have a structure that is much more complex than the threshold structure of state-independent utilities; see Example \ref{example2}. The results of \cite{babichenko2022regret} indicate that the sender can guarantee low regret whenever the receiver's utility is monotonic in the state. Our results provide another instance where low regret can be guaranteed.

\section{Preliminaries: Known Value of the Seller}
\label{model_section}
Before introducing the case in which the seller's value is unknown, we briefly discuss the model of  \cite{bergemann2015} in which the seller's value $\us\in \R_{+}$ is known by the market designer. 
Let $B \coloneqq \set{b_1, ... b_n} \subset \R_{+}$ be the buyers' valuations.
We assume that $0 < b_1 < ... < b_n$.
Let $\mu \in int(\Delta(B))$ be the prior buyer distribution.\footnote{For any set $A$, we denote by $int(A)$ the interior of $A$, and $\Delta(A)$ is the set of all probability distributions over $A$.} A segmentation of the market designer is a Bayes plausible posterior distribution $\sigma$, i.e., the set of all possible segmentation is given by:

\begin{equation*}
    \Sigma \coloneqq \set{\sigma \in \Delta(\Delta(B)) | \E_{p \sim \sigma}[p] = \mu}
\end{equation*}

We assume that when indifferent, the seller sets the lowest price. Also, when the buyer is indifferent (i.e., the price equals its valuation) the buyer buys the product. Thus, for a given posterior $p \in \Delta(B)$, the seller's price is given by:

\begin{equation*}
    \pi(p;\us) \coloneqq \min{\argmax_{b_i \in B} (b_i - \us) \cdot {\sum_{j=i}^n p_j}}
\end{equation*}
We denote by $b_{i^*}$ the \emph{monopolistic price} that is the price that will be set without any information; i.e., $b_{i^*} = \pi(\mu;\us)$.
We consider a market designer who aims to maximize the buyer's surplus. For a given posterior $p \in \Delta(B)$, the buyer's surplus is

\begin{equation*}
    U(p;\us) \coloneqq {\sum_{j=1}^n p_j \cdot \max \{b_j - \pi(p;\us), 0 \}}
\end{equation*}
and the buyer surplus for a given segmentation $\sigma \in \Sigma$ is simply the expectation over the possible posteriors:

\begin{equation*}
    U(\sigma,\us) \coloneqq \E_{p \sim \sigma}[U(p;\us)] = \sum_{p \in \supp{\sigma}} \sigma(p) \cdot {\sum_{j=1}^n p_j \cdot \max \{b_j - \pi(p;\us), 0 \}}
\end{equation*}
where the summation over $p$ should be replaced with integration if $\supp{\sigma}$ is an uncountable set.
We denote by $U^*(\us)$ the optimal buyers surplus (across all segmentations), and by $\Sigma^*(\us)$ the set of optimal market segmentations, as the optimal segmentation may not be unique. \cite{bergemann2015} provide a very clean formula for $U^*(\us)$:

\begin{equation}
\label{eq: optimal surplus}
    U^*(\us) = \sum_{j=1}^n \mu_j \cdot \max \{b_j - \us,0\} - (b_{i^*} - \us) \cdot \sum_{j=i^*}^n \mu_j
\end{equation}

This formula has a clean interpretation. The first term captures the maximal social welfare. The second term is the monopolistic surplus of the seller. Notice that the monopolistic surplus can be guaranteed by the seller for every segmentation $\sigma$ (simply by ignoring the information). The sum of surpluses (buyer plus seller) cannot exceed the social welfare. From these trivial arguments, we deduce that the expression of Equation \eqref{eq: optimal surplus} is an upper bound on the buyer's surplus. The surprising result of \cite{bergemann2015} shows that this bound can be reached by a careful choice of segmentation. 

The behavior of the optimal surplus $U^*$ as a function of the seller's value for the product $\us$ will play a significant role in our analysis. In Lemma \ref{key_prop_lemma} (whose proof is relegated to Appendix \ref{app:l1}) we summarize the key properties that will be utilized.

\begin{lemma}
    \label{key_prop_lemma}
    For every prior buyer distribution $\mu$ the optimal buyer surplus function $U^*(\cdot)$ is weakly decreasing, absolutely continuous, nonnegative, and differentiable up to a finite number of points.
    Moreover, there exists $s^* \ge 0$ such that $U^*(\us)$ is constant over $[0,s^*]$, strictly decreasing over $[s^*,b_{n-1}]$, and constantly 0 over $[b_{n-1},\infty)$.
\end{lemma}

\section{Robust Price Discrimination}
\label{robust_subsection}
From now on, we consider the \emph{robust approach} for studying the price discrimination model with uncertainty regarding the seller's valuation. 
Departing from the standard model of \citet{bergemann2015}, we assume that $\us$ is unknown to the market designer, who chooses a segmentation $\sigma$.\footnote{In settings that admit private information for the seller, an alternative approach is to allow the designer to choose a menu of segmentations, from which the seller chooses one. This approach has been adopted by \cite{SYnew} when the private information of the seller comes in the form of partial information about the buyer's valuation.}
The designer’s objective is to minimize the \emph{regret}, defined as the maximal disparity between the optimal buyer surplus (had 
$\us$ been known) and the surplus generated by the segmentation, maximized over all possible values of $\us$.\footnote{Another natural model that might be considered is the Bayesian one: $\us$ is drawn according to a commonly known distribution $\us\sim F$. Notably, this Bayesian model boils down to a standard Bayesian persuasion problem \citep{kamenica2011bayesian}. We elaborate on the Bayesian model in Section \ref{binary_type_section}, as it plays a crucial role in our analysis, particularly in showing that our regret bound is tight in the case of binary buyer valuation.}
Formally, the regret of a given segmentation is defined as follows:

\begin{equation*}
    R(\sigma) \coloneqq \max_{\us}\{U^*(\us) - U(\sigma,\us)\}
\end{equation*}

and the \emph{overall regret} is defined as the minimal regret that can be achieved by any market segmentation:

\begin{equation}
\label{standard_reg_eq}
    R \coloneqq \min_{\sigma \in \Sigma} {R(\sigma)} 
\end{equation}

It is worth highlighting that finding a closed form of the regret-minimizing market segmentation seems to be a very challenging problem: the outer minimization problem defined by Equation \eqref{standard_reg_eq} is intractable since the minimization is taken over an infinite-dimensional space $\Sigma$ (notice that, unlike in the standard model of \citet{bergemann2015}, we can not restrict attention to segmentations with finite support). 

Our goal, then, is to find a simple and intuitive market segmentation that is independent of $\us$ and achieves low regret. That is, we want to find an information revelation policy of the market designer that is robust to the seller's valuation $\us$, and achieves a near-optimal buyer surplus regardless of it. 
Despite the technical difficulties in solving the regret-minimization problem, we are able to introduce such a robust market segmentation of the designer, and show that the overall regret is bounded from above by $U^*(0)/e$. Then, we also show that this bound is tight for the case of binary buyer type, i.e. when $n=2$. This robust market segmentation takes the following form: the market designer first draws $s_D$, and then it applies the optimal segmentation of \cite{bergemann2015} \emph{as if} $\us = s_D$. We begin by formally defining the class of segmentations that contains our robust market segmentation for the market designer:

\begin{definition}\label{def:BBM}
    A BBM market segmentation is a strategy of the market designer $\sigma \in \Sigma$, for which there exists a random variable $s_D$ such that for every posterior $p \in \Delta(B)$, $p$ is drawn with probability $\sigma(p) = \E [\sigma^*(s_D)(p)]$, where $\sigma^*(s_D) \in \Sigma^*(s_D)$ is some optimal segmentation.
\end{definition}

\begin{remark}
    While the optimal segmentation $\sigma^*(s_D)$ may not be unique, all of our theoretical results (Theorem \ref{main_res_theorem} and Theorem \ref{main_res_bintype_theorem}) are agnostic to the choice of $\sigma^*(s_D)$. However, as we demonstrate in Example \ref{example:unif3_cont}, the implementation of the optimal segmentation may affect the actual regret that can be achieved by such a BBM segmentation.
\end{remark}

When $s_D$ is a continuous random variable, the BBM segmentation can be identified with the corresponding density function $g$. In this case, the unconditional posterior distribution can be written as $\sigma(p) = \int_{s_D=0}^{b_n} g(s_D) \cdot \sigma^*(s_D)(p) \cdot ds_D$. Now, the following theorem introduces our main result:

\begin{theorem}
    \label{main_res_theorem}
    The overall regret is bounded from above by $\frac{U^*(0)}{e}$, and this bound is attained by a BBM market segmentation.
\end{theorem}

The proof of the theorem relies on the analysis of a specific class of zero-sum games with a continuum of actions. The class of games and the solution is summarized in the following lemma, whose proof is relegated to Appendix \ref{app:l2}.

\begin{lemma}
    \label{main_tech_lemma}
    Let $0<\alpha<\beta$, and let $u:[0,\beta]\to \mathbb{R}_+$ be an absolutely continuous function that is constant in $[0,\alpha]$ and strictly decreasing in $[\alpha,\beta]$. Let $v$ be a two-player zero-sum game in which players 1,2 choose real numbers $x,y\in [0,\beta]$ (correspondingly). The utility of Player 1 is given by
        \begin{equation*}
        v(x,y) = \left\{\begin{array}{lr}
            u(x), & \text{for } x > y \\
            u(x)-u(y), & \text{for } x \le y
            \end{array}\right\}
        \end{equation*}
    
    Then, the value of the game is $\frac{u(0)}{e}$, and it can be guaranteed to Player 2 by playing a mixed strategy with the following density function:
    \begin{equation*}
        g(y) = \left\{\begin{array}{lr}
        -\frac{u'(y)}{u(y)}, & \text{for } 
            \alpha \le y \le \delta \\
            0, & \text{otherwise}
        \end{array}\right\}
    \end{equation*}
    for $\delta \in (\alpha,\beta)$ such that $u(\delta) = \frac{u(0)}{e}$.\footnote{The monotonicity of $u$ implies that $u$ is differentiable almost everywhere, and therefore the density function $g$ is defined almost everywhere.}
\end{lemma}

Interestingly, the family of zero-sum games considered by Lemma \ref{main_tech_lemma} extends the zero-sum game analyzed by \cite{bergemann2008pricing,bergemann2011robust} to derive a robust monopoly pricing strategy. In their setup, a monopolist sets a price $p \in [0,1]$ without knowing the buyer's valuation $b \in [0,1]$, and aims to minimize the regret, which can be written as $b - p$ if $p \le b$, and $b$ otherwise. The resulting zero-sum game is a private case of ours, for the case where $u(t) = 1-t$.\footnote{To see the mapping between the games, simply set $x = 1-b$ and $y= 1-p$.} From this perspective, our result can be seen as extending the technique from the one-dimensional robust monopoly pricing problem (the seller chooses a \emph{price}) to an infinite-dimensional problem of robust price discrimination (the designer chooses a \emph{segmentation}).

Notice that the value of the game in Lemma \ref{main_tech_lemma} depends on the initial condition $u(0)$ and does not depend on the behavior of the function $u$ besides the monotonicity property. This property is very surprising because the utilities in the game depend on the behavior of $u$ on the entire interval $[0,\beta]$ and in equilibrium, both players are playing actions in $[\alpha,\beta]$ with positive probability. Intuitively, this property follows from the particular additive structure of the payoffs and the ability of the players to adjust their mixed strategy to "cancel out" the dependence on the particular behavior of $u$ in $[\alpha,\beta]$. Moreover, for both players, it is optimal to cancel out this dependence.
We now turn to prove Theorem \ref{main_res_theorem}:

\begin{proof}[\textbf{Proof of Theorem \ref{main_res_theorem}.}]

Consider a zero-sum game between the market designer and an adversary, in which the adversary chooses $\us$ and the market designer chooses $s_D$, and then plays $\sigma^*(s_D)$. The utility of the adversary in this zero-sum game is defined to be the difference between the optimal buyer surplus $U^*(\us)$, and the actual buyer surplus, $U(\sigma^*(s_D),\us)$. Trivially, the value of this game is an upper bound on the overall regret, since the market designer is forced to play a BBM segmentation. Denote the utility function of the adversary in the auxiliary game by $v$, and by definition it holds that:

\begin{equation}
\label{eq_thrm_1_orig_utils}
    v(\us,s_D) = U^*(\us) - U(\sigma^*(s_D),\us)
\end{equation}

Now, notice that from the non-negativity of the buyer surplus, it holds that $v(\us,s_D) \le U^*(\us)$ for any $\us$ and $s_D$. Moreover, for any fixed segmentation $\sigma$, the buyer surplus $U(\sigma,\us)$ is non-increasing as a function of $\us$. 
Therefore, $\us \le s_D$ implies that $U(\sigma^*(s_D),\us) \ge U(\sigma^*(s_D),s_D) = U^*(s_D)$, which means that $v(\us,s_D) \le U^*(\us) - U^*(s_D)$. Altogether, for any $\us$ and $s_D$, it holds that $v(\us,s_D)$ is bounded from above by the following function:

\begin{equation}
\label{eq_thrm_1_relaxed_utils}
    \tilde{v}(\us,s_D) = \left\{\begin{array}{lr}
        U^*(\us), & \text{for } \us > s_D \\
        U^*(\us)-U^*(s_D), & \text{for } \us \le s_D
        \end{array}\right\}
\end{equation}

Therefore, the value of the game defined by $\tilde{v}$ is an upper bound of the value of the game defined by $v$, and hence it also bounds the overall regret from above. 

Note that for the adversary, playing $\us>b_{n-1}$ in the game defined by $\tilde{v}$ yields utility $0$ regardless of the market designer's strategy $s_D$, hence it is weakly a dominated strategy (e.g. by $b_{n-1}$).
Now, for the market designer, playing $s_D>b_{n-1}$ is equivalent to $b_{n-1}$ since $s_D \ge \us$ and $U^*(s_D)=0$. Therefore, for the purpose of finding the value of the game, it can be assumed without loss of generality that both players play $s_D > b_{n-1}$ and $\us > b_{n-1}$ with probability zero.

Now, Lemma \ref{key_prop_lemma} implies that the game defined by $\tilde{v}$ (after strategies elimination) satisfies the conditions of Lemma \ref{main_tech_lemma}. Hence, the value of the game defined by $\tilde{v}$ is $\frac{U^*(0)}{e}$, and therefore this is an upper bound on the overall regret that can be attained by a BBM segmentation of the market designer.

\end{proof}

Notice that Lemma \ref{main_tech_lemma} enables the construction of a concrete market segmentation that guarantees this upper bound on the regret since the market designer corresponds to Player 2 in Lemma \ref{main_tech_lemma}. The optimal strategy of the designer can be interpreted through the lens of \emph{survival analysis}. Using the notation of the lemma, one can think of the non-increasing function $u(t)$ as a \emph{survival function}, specifying the survival probability beyond time $t$.\footnote{For this interpretation, one can assume that $u$ is normalized such that $u(0)=1$.}
The designer's optimal strategy is then to draw $s_D$ from a distribution whose density corresponds to the \emph{hazard function} for this survival function, defined as:

\begin{equation*}
    h(t) = -\frac{u'(t)}{u(t)}
\end{equation*}

Now, recall that $u(s)$ represents the optimal surplus function $U^*(s)$, so the hazard function quantifies the instantaneous "decay rate" of the optimal surplus relative to the remaining potential surplus at any valuation $s$. Intuitively, it captures how the designer adjusts the segmentation strategy to mitigate loss in buyer surplus as $s$ increases.
This "hazard strategy", which turns out to be the optimal among all BBM strategies in the relaxed zero-sum game, can be seen as ensuring that the "survival probability" of the buyer's optimal surplus is maximally preserved, even under adversarial conditions. In this sense, the designer segments the market as if it were managing the survival of surplus over the space of seller valuations.

\section{Tightness in the Case of Binary Buyer Type}
\label{binary_type_section}

We now turn to analyze the special case in which there are only two possible buyer types, namely $n=2$. We show that in this case, the upper bound on the overall regret obtained in Theorem \ref{main_res_theorem} is tight, by showing a mixed strategy of the adversary that guarantees a regret of at least $\frac{U^*(0)}{e}$, and combined with the upper bound presented in Theorem \ref{main_res_theorem} we conclude that the overall regret is precisely $\frac{U^*(0)}{e}$.

Let $b_2$ and $b_1$ be the two possible buyer valuations. Without loss of generality, we assume that $b_2 - b_1 = 1$.\footnote{This is without loss of generality since both the regret and the optimal buyer surplus function are linear in the difference between the two buyer valuations.} For any posterior $p \in \Delta(B)$ we identify $p$ with $\prob_p(b_1) = p_1$. We note that now the seller's optimal price and the buyer surplus function take the following simpler form for any given posterior $p$ and seller valuation $\us$:

\begin{equation*}
    \pi(p;\us) \coloneqq \left\{\begin{array}{lr}
        b_1, & \text{for }b_1 - \us \ge (b_2 - \us)(1 - p)\\
        b_2, & \text{for }b_1 - \us < (b_2 - \us)(1 - p)
        \end{array}\right\}
\end{equation*}

\begin{equation*}
    U(p;\us) \coloneqq \left\{\begin{array}{lr}
        1-p, & \text{for }b_1 - \us \ge (b_2 - \us)(1 - p)\\
        0, & \text{for }b_1 - \us < (b_2 - \us)(1 - p)
        \end{array}\right\}
\end{equation*}

In the following technical Lemma, we use the result of \cite{bergemann2015} to obtain a closed form of the optimal buyers surplus when $\us$ is known to the market designer:

\begin{lemma}
\label{bbm_lemma}
    When $n=2$ and $\us$ is known to the market desginer, the optimal buyers surplus is given by: 
    \begin{equation*}
        U^*(\us) = \left\{\begin{array}{lr}
            1 - \mu, & \text{for } \us < b_2 - \frac{1}{\mu}
             \\
            
            (b_1 - \us) \mu, & \text{for } b_2 - \frac{1}{\mu} \le \us \le b_1 \\
            
            0, & \text{for } \us > b_1 \end{array}
        \right\}
    \end{equation*}
\end{lemma}

\begin{figure}[t]
\centering
\begin{tikzpicture}
    \definecolor{darkgreen}{rgb}{0.0, 0.5, 0.0} 
    \begin{axis}[
        axis lines = left,
        xlabel = $s$,
        xtick={0,1/2,1},
        xticklabels={0,$b_2 - \frac{1}{\mu}$,$b_1$},
        ytick=\empty, 
        ymax = 0.4,
        thick,
        legend pos=north east,
        legend cell align={left} 
    ]
    \addplot [domain=0:0.5, samples=100, color=darkgreen, thick, forget plot] {1/3};
    \addplot [domain=0.5:1, samples=100, color=darkgreen, thick, forget plot] {2*(1-x)/3};
    \addplot [domain=1:1.3, samples=100, color=darkgreen, thick] {0};
    
    \addlegendentry{$U^*(s)$}
    \end{axis}
\end{tikzpicture}
\caption{The optimal buyer surplus function in the case of binary buyer type.}
\label{fig_opt_surplus}
\end{figure}

The proof of Lemma \ref{bbm_lemma} is relegated to Appendix \ref{app:3}. Figure \ref{fig_opt_surplus} visualizes the optimal buyer surplus function in the case of binary buyer type. Using this technical Lemma we can now conclude that the overall regret in the binary buyer type case is precisely $\frac{U^*(0)}{e}$:

\begin{theorem}
\label{main_res_bintype_theorem}
    When $n=2$, the overall regret is $\frac{U^*(0)}{e}$, and it is obtained by a BBM market segmentation.
\end{theorem}

In order to prove Theorem \ref{main_res_bintype_theorem} we make use of the min-max theorem of \citet{sion1958general}, and write the overall regret in a Bayesian manner. As a first step, let us consider an alternative model for price discrimination with incomplete information about the seller's valuation---the \emph{Bayesian} model, in which we assume that $\us$ is drawn from a commonly known prior distribution $F$, where $F$ is the CDF of the prior distribution. The segmentation problem now can be viewed as a standard Bayesian persuasion problem as introduced by \cite{kamenica2011bayesian}. The unknown state is the buyer's valuation $b$ which is drawn from a common prior distribution $\mu$. The market designer (the \emph{sender} in persuasion) knows the state and chooses a segmentation (a \emph{signaling policy} in persuasion).
The seller is the receiver.
The market designer's utility as a function of the seller's posterior (the \emph{indirect utility}) is given by

\begin{equation*}
    u_F(p)=\mathbb{E}_{\us\sim F}[U(p,\us)]
\end{equation*}

Building upon \cite{aumann1995repeated},  \cite{kamenica2011bayesian} elegantly characterize the solution of this persuasion problem via the notion of concavification, which is denoted by $cav$.\footnote{The concavification of $u$ is defined to be the minimal concave function that is pointwise above $u$.} The optimal expected utility of the market designer is $cav(u_F)(\mu)$. Moreover, a segmentation $\sigma \in \Sigma$ is optimal if and only if $(\mu, cav(u_F)(\mu))$ is a convex combination of $(p, u_F(p))_{p \in supp(\sigma)}$, with weights corresponding to $(\sigma(p))_{p \in supp(\sigma)}$.

In the case of $n=2$ we can easily derive a closed form of the designer's indirect utility function.
Denote by $t(p)$ the threshold $\us$ for which the seller is indifferent between prices $b_1$ and $b_2$, when the posterior belief is $p = \prob(b_1)$. It holds that:

\begin{equation*}
    b_1 - t(p) = (b_2 - t(p)) \cdot (1-p) \Rightarrow t(p) = \frac{b_1 - b_2 \cdot (1-p)}{p} = b_2 - \frac{1}{p}
\end{equation*}

Now, notice that the buyer surplus is non-zero if and only if $\us \le t(p)$ and $b = b_2$. In this case, the surplus is $b_2 - b_1$. When $\us \ge t(p)$ and $b = b_1$ there is no trade, hence zero surplus. When $\us \le t(p)$ and $b = b_1$ (or, symmetrically, when $\us \ge t(p)$ and $b = b_2$) there is trade, but zero surplus $b_1 - b_1$ (or $b_2 - b_2$). Therefore, the indirect utility of the market designer is given by:

\begin{equation*}
    u(p) = (b_2 - b_1) \cdot (1-p) \cdot F(t(p)) = (1-p) \cdot F(t(p))
\end{equation*}
where the last equality is due to the assumption of $b_2 - b_1 = 1$.

Now, going back to our regret-minimization framework, we observe that the overall regret can be written in terms of the indirect utility concavification (as in the Bayesian model described above):

\begin{equation}
\label{bayesian_reg_eq}
    R = \sup_{F}\{\E_{\us \sim F}[U^*(\us)] - cav(u_F)(\mu)\}
\end{equation}

Equation \eqref{bayesian_reg_eq} follows from the min-max theorem of \cite{sion1958general}.\footnote{Notice that our instance falls into the min-max theorem of \cite{sion1958general}, as the objective is linear in $\sigma$ and lower-semi continuous in the adversary’s mixed strategy. This is because $U(\sigma,\us)$ is upper-semi continuous in $\us$, $U^*(\us)$ is continuous in $\us$, and $U(\sigma,\us)$ is linear in $\sigma$.} We view the interaction as a zero-sum game between the market designer (who chooses a segmentation) and the adversary (who chooses $\us$). By the min-max theorem, there exists a mixed strategy of the adversary $F$ that guarantees the value of the game $R$ where $cav(u_F)(\mu)$ is the best reply of the market designer against the mixed strategy $F$.

The above formulation of the overall regret will now be used to prove Theorem \ref{main_res_bintype_theorem}. In particular, we will show that in the case of binary type case, the adversary has a strategy $F$ in the zero-sum game that guarantees a regret of at least $\frac{U^*(0)}{e}$, which proves the tightness of the bound obtained in Theorem \ref{main_res_theorem}.

\begin{proof}[\textbf{Proof of Theorem \ref{main_res_bintype_theorem}.}]

Consider a zero-sum game between the market designer and an adversary, in which the market designer selects a market segmentation $\sigma$ and the adversary selects $\us$ to maximize the regret. Note that unlike the game defined in the proof of Theorem \ref{main_res_theorem}, in this game the market designer is not restricted to BBM segmentations, and can choose any arbitrary market segmentation $\sigma \in \Sigma$. Denote the value of this game by $R$, and note that the value of this game is the overall regret. From Theorem \ref{main_res_theorem}, $R \le \frac{U^*(0)}{e}$. It is therefore left to show that the adversary has a mixed strategy that guarantees a utility of at least $\frac{U^*(0)}{e}$.

For each $\beta \in [b_2 - \frac{1}{\mu}, b_1)$, define a distribution $F_\beta$ such that $supp(F_\beta) = [0, \beta]$:

    \begin{equation*}
        \forall t \in [0, \beta]: F_\beta(t) = \frac{b_1 - \beta}{b_1 - t}
    \end{equation*}

Notice that the distribution has an atom at $x = 0$, and the corresponding density function is given by $f_\beta(t) = \frac{b_1 - \beta}{(b_1 - t)^2}$. Now, under the assumption that $\us \sim F_\beta$, the indirect utility of the market designer, as a function of the belief $p = \prob_p(b_1)$, is given by:

    \begin{equation*}
    u_\beta(p) = (1 - p) F_\beta(t(p))
    \end{equation*}
    
    Plugging in $F_\beta$, we obtain:
    
\begin{equation*}
    u_\beta(p) = \left\{\begin{array}{lr}
        0, & \text{for } p < \frac{1}{b_2}\\
        (b_1 - \beta) p, & \text{for } \frac{1}{b_2} \le p \le \frac{1}{b_2 - \beta}\\
        1 - p, & \text{for } p > \frac{1}{b_2 - \beta}
        \end{array}\right\}
\end{equation*}

The concavification of the indirect utility has the following form:

    \begin{equation*}
        cav(u_\beta)(p) = \left\{\begin{array}{lr}
            (b_1 - \beta) p, & \text{for } p \le \frac{1}{b_2 - \beta}\\
            1 - p, & \text{for } p \ge \frac{1}{b_2 - \beta}
            \end{array}\right\}
    \end{equation*}

\begin{figure}[t]
\centering
\begin{tikzpicture}
    \begin{axis}[
        axis lines = left,
        xlabel = $p$,
        xtick={0,0.25,0.333,1},
        xticklabels={0,$\frac{1}{b_2}$,$\frac{1}{b_2 - \beta}$,1},
        ytick=\empty, 
        thick,
        legend pos=north east,
        legend cell align={left} 
    ]
    
    \addplot [domain=0:1/4, samples=100, color=blue, thick, forget plot] {0};
    \addplot [domain=1/4:1/3, samples=100, color=blue, thick, forget plot] {2*x};
    \addplot [domain=1/3:1, samples=100, color=blue, thick] {1-x};
    \addlegendentry{$u_{\beta}(p)$} 
    
    \addplot [domain=0:1/3, samples=100, color=orange, thick, dashed, forget plot] {2*x}; 
    \addplot [domain=1/3:1, samples=100, color=orange, thick, dashed] {1-x};
    \addlegendentry{$cav(u_{\beta})(p)$} 

\end{axis}
\end{tikzpicture}
\caption{The indirect utility of the market designer $u_{\beta}(p)$ and its concavification, for a given distribution $F_\beta$ over $\us$.}
\label{fig2}
\end{figure}
 
 Note that the concavification has the structure of a triangle (see Figure \ref{fig2}): for $p \le \frac{1}{b_2 - \beta}$ it is a linear function with slope $(b_1 - \beta)$, and for $p > \frac{1}{b_2 - \beta}$ it is a linear function with slope $-1$. 
 Using the distribution family $\{F_\beta\}_{\beta < b_1}$ combined with Equation \eqref{bayesian_reg_eq}, we can obtain the following lower bound on the overall regret:
    
    \begin{equation}
    \label{lb_eq}
        R \ge \sup_{b_2 - \frac{1}{\mu} \le \beta < b_1}R(F_\beta)
    \end{equation}

Since we chose $\beta \ge b_2 - \frac{1}{\mu}$, we know that $cav(u_\beta)(\mu) = (b_1 - \beta) \mu$. It is now left to compute $\E_{\us \sim F_\beta}[U^*(\us)]$, and here we distinguish between two cases:

\textbf{First case: $\mu \le \frac{1}{b_2}$.} In this case, we know that $b_2 - \frac{1}{\mu} \le 0$. Therefore, for every $\us$ in the support of $F_\beta$, it holds that $\us \ge 0 \ge b_2 - \frac{1}{\mu}$, and Lemma \ref{bbm_lemma} implies that $U^*(\us) = (b_1-\us) \mu$. Now,

    \begingroup
    \allowdisplaybreaks
    \begin{align*}
        \E_{\us \sim F_\beta}[U^*(\us)] &= F_\beta(0) \cdot U^*(0) + \int_{\us=0}^{\beta} U^*(\us) f_\beta (\us) d\us \\
        &= (b_1 - \beta) \mu + \int_{\us=0}^{\beta} U^*(\us) f_\beta (\us) d\us \\ 
        &= (b_1 - \beta) \mu + \mu \int_{\us=0}^{\beta} (b_1 - \us) \cdot \frac{b_1 - \beta}{(b_1 - \us)^2} d\us \\ 
        &= (b_1 - \beta) \mu + \mu \int_{\us=0}^{\beta} \cdot \frac{b_1 - \beta}{b_1 - \us} d\us \\ 
        &= (b_1 - \beta) \mu \bigg( 1 + \ln(b_1) - \ln(b_1 - \beta) \bigg) \\
        &= (b_1 - \beta) \mu \bigg( 1 + \ln\bigg(\frac{b_1}{b_1 - \beta}\bigg)\bigg)
    \end{align*}
    \endgroup

\textbf{Second case: $\mu > \frac{1}{b_2}$.} In this case, Lemma \ref{bbm_lemma} implies that for  $\us \in [0, b_2 - \frac{1}{\mu}]$, the optimal surplus is $U^*(\us) = 1 - \mu$, and otherwise $U^*(\us) = (b_1 - \us) \mu$. In this case, we obtain:

    \begingroup
    \allowdisplaybreaks
    \begin{align*}
        \E_{\us \sim F_\beta}[U^*(\us)] &= \int_{\us=0}^{\beta} U^*(\us) f_\beta (\us) d\us \\
        &= \int_{\us=0}^{b_2 - \frac{1}{\mu}} U^*(\us) f_\beta (\us) d\us + \int_{\us=b_2 - \frac{1}{\mu}}^{\beta} U^*(\us) f_\beta (\us) d\us \\
        &= (1-\mu) F_\beta\bigg(b_2 - \frac{1}{\mu}\bigg) + \mu \int_{\us=b_2 - \frac{1}{\mu}}^{\beta} (b_1 - \us) \cdot \frac{b_1 - \beta}{(b_1 - \us)^2} d\us \\
        &= (1-\mu) F_\beta\bigg(b_2 - \frac{1}{\mu}\bigg) + \mu \int_{\us=b_2 - \frac{1}{\mu}}^{\beta} \frac{b_1 - \beta}{b_1 - \us} d\us \\
        &= (1-\mu) F_\beta\bigg(b_2 - \frac{1}{\mu}\bigg) + \mu (b_1 - \beta) \bigg( \ln\bigg(b_1 - b_2 + \frac{1}{\mu}\bigg) - \ln(b_1 - \beta) \bigg) \\
        &= (1-\mu) F_\beta\bigg(b_1 - \frac{1 - \mu}{\mu}\bigg) + \mu (b_1 - \beta) \bigg( \ln\bigg(\frac{1 - \mu}{\mu}\bigg) - \ln(b_1 - \beta) \bigg) \\
        &= \mu (b_1 - \beta) + \mu (b_1 - \beta)\ln\bigg(\frac{1 - \mu}{\mu(b_1 - \beta)}\bigg)
    \end{align*}
    \endgroup

Altogether, the regret with respect to $F_\beta$ as a function of $\mu$ is given by:

    \begin{equation*}
        R(F_\beta) = \left\{\begin{array}{lr}
            \mu (b_1 - \beta) \ln\big(\frac{b_1}{b_1 - \beta}\big), & \text{for } \mu \le \frac{1}{b_2}\\
            \mu (b_1 - \beta)\ln\big(\frac{1 - \mu}{\mu(b_1 - \beta)}\big), & \text{for } \mu > \frac{1}{b_2}
            \end{array}\right\}
    \end{equation*}

Now, let $\beta^* = b_1 \bigg( 1 - \frac{1}{e} \bigg)$ if $\mu \le \frac{1}{b_2}$, otherwise $\beta^* = b_1 - \frac{1 - \mu}{e\cdot \mu}$. Plugging into Equation \eqref{lb_eq} we obtain the following lower bound:

    \begin{equation*}
        R \ge \sup_{b_2 - \frac{1}{\mu} \le \beta < b_1}R(F_\beta) \ge R(F_{\beta^*}) = \left\{\begin{array}{lr}
            \frac{b_1 \mu}{e}, & \text{for } \mu \le \frac{1}{b_2}\\
            \frac{1-\mu}{e}, & \text{for } \mu > \frac{1}{b_2}
            \end{array}\right\} = \frac{U^*(0)}{e}
    \end{equation*} 

\end{proof}

The result is somewhat surprising. To see why, consider the following example in a Bayesian price discrimination model:

\begin{example}
\label{example2}
    Consider the case where $b_2 = 4, b_1 = 3$ and $\us \sim Uni([0,1] \cup [2 \frac{1}{2},3 \frac{1}{2}])$. Note that $t(p) = 4 - \frac{1}{p} \in [0, 1]$ if and only if $p \in [\frac{1}{4}, \frac{1}{3}]$, and $t(p) \in [2 \frac{1}{2},3 \frac{1}{2}]$ if and only if $p \ge \frac{2}{3}$. Therefore, the indirect utility of the market designer is given by:

    \begin{equation*}
            u(p) = \left\{\begin{array}{lr}
            0, & \text{for } p \in [0,\frac{1}{4}]\\
            \frac{(1-p)(4p-1)}{2p}, & \text{for } p \in [\frac{1}{4},\frac{1}{3}]\\
            \frac{1-p}{2}, & \text{for } p \in [\frac{1}{3},\frac{2}{3}]\\
            \frac{1-p}{2} \cdot (4 - \frac{1}{p} - \frac{3}{2}), & \text{for } p \in [\frac{2}{3},1]\\
            \end{array}\right\}
    \end{equation*}

    Figure \ref{fig_example_2} visualizes the indirect utility and its concavification, and it can be now seen that for certain priors $\mu$, the unique optimal segmentation is not a BBM segmentation. For instance, if $\mu = \frac{2}{3}$, the optimal segmentation is a mixture of two mixed posterior, which is different than the optimal segmentation as in \cite{bergemann2015}, hence cannot be induced by any distribution over $s_D$.
\end{example}

\begin{figure}[t]
\centering
\begin{tikzpicture}
    \begin{axis}[
        axis lines = left,
        xlabel = $p$,
        xtick={0,0.25,0.33,0.667,1},
        xticklabels={0,$\frac{1}{4}$,$\frac{1}{3}$,$\frac{2}{3}$,1},
        ytick=\empty, 
        thick,
        legend pos=north east,
        legend cell align={left} 
    ]
    
    \addplot [domain=0:1/4, samples=100, color=blue, thick, forget plot] {0};
    \addplot [domain=1/4:1/3, samples=100, color=blue, thick, forget plot] {((1-x)*(4*x-1))/(2*x)};
    \addplot [domain=1/3:2/3, samples=100, color=blue, thick, forget plot] {(1-x)/2};
    \addplot [domain=2/3:1, samples=100, color=blue, thick] {(1-x)/2 * (4 - 1/x - 3/2)};
    \addlegendentry{$u(p)$} 
    
    \addplot [domain=0:1/3, samples=100, color=orange, thick, dashed, forget plot] {x};
    \addplot [domain=1/3:0.85, samples=100, color=orange, thick, dashed, forget plot] {1/3-(x-1/3)/2.25};
    \addplot [domain=0.85:1, samples=100, color=orange, thick, dashed] {(1-x)/2 * (4 - 1/x - 3/2)};
    \addlegendentry{$cav(u)(p)$} 
    
    \end{axis}
\end{tikzpicture}
\caption{The indirect utility of the market designer and its concavification, corresponding to Example \ref{example2}.}
\label{fig_example_2}
\end{figure}

Example \ref{example2} demonstrates that an optimal market segmentation is not necessarily a BBM segmentation.
However, Theorem \ref{main_res_bintype_theorem} shows that at least in the binary buyer type case, the adversary has an optimal strategy for which the designer's best reply is indeed a BBM segmentation. That is, in this case there always exists a regret-minimizing segmentation, which is a BBM segmentation.

\subsection{Improved Regret Beyond the Case of Binary Buyer Type}

While our BBM segmentation approach is regret-minimizing for the case of binary buyer valuation, it might be suboptimal when applied to instances with more than two valuations, as illustrated by the following example.

\begin{example}
\label{example:unif3}
Consider the case where the prior $\mu$ is uniform over the buyer valuations $\{1, 2, 3\}$. To evaluate the tightness of the theoretical regret bound $\frac{U^*(0)}{e}$, we numerically solve the corresponding zero-sum game between the adversary (who chooses the seller's valuation $s$) and the designer (who chooses a segmentation $\sigma$, not restricted to BBM). 
We discretize the strategy spaces of both players and formulate the interaction as a linear program (LP). Solving the LP yields an approximate game value (corresponding to the overall regret) of $\approx 0.165$, which is significantly lower than the theoretical upper bound $\frac{U^*(0)}{e} \approx 0.245$. 
The resulting optimal (approximated) regret-minimizing segmentation puts an aggregated mass of $\approx 0.1$ on segments of the form $(p,0,1-p)$ for some $p \in (0,1)$, which cannot be included in any BBM segmentation.\footnote{The reason is that for $(p,0,1-p)$ to be included in a BBM segmentation, there must exist a seller valuation $s$ for which the seller is indifferent between setting the monopolistic price (which can only be either $b_2=2$ or $b_3=3$) and the lowest price in the segment, which is $b_1=1$. This can only happen if the monopolistic price is $b_3$ (as the segment puts zero mass on $b_2$), but this implies that $s > 1$, in contradiction to $b_1$ being an optimal price in $(p,0,1-p)$.}
\end{example}

We recall that in the definition of a BBM segmentation (Definition \ref{def:BBM}), after drawing $s_D$, the designer chooses a buyer-optimal segmentation with respect to the "guessed" seller value $s_D$. However, such a buyer-optimal segmentation is not necessarily unique. Example \ref{example:unif3} demonstrates that the choice of the buyer-optimal segmentation might affect the regret it actually achieves.

One buyer-optimal segmentation suggested by \cite{bergemann2015} is the \emph{equal revenue} segmentation in which, within each segment, the seller is indifferent between all the prices in the support of this segment. As shown by \cite{bergemann2015}, such an optimal segmentation always exists. 
A key observation is that in an equal revenue segment (w.r.t $s_D$), if $s>s_D$, then the seller chooses the maximal price in the support, and consequently, the buyer's revenue is 0. This is captured by the following lemma, whose proof is deferred to Appendix \ref{app:proof_of_lemma_eq_rev}.

\begin{lemma}
    \label{lemma:equal_revenue_segments}
    Let $p = (p_1, ... p_n) \in \Delta(B)$ be an equal revenue segment with respect to seller's valuation $s_D$; that is, for all $i, j \in [n]$ it holds that:
    \begin{equation*}
        (b_i-s_D) \sum_{k=i}^n p_k = (b_j-s_D) \sum_{k=j}^n p_k
    \end{equation*}
    Then, it holds that:
    \begin{equation*}
        \pi(p;s) = \left\{\begin{array}{lr}
            \max \supp (p), & \text{if } s > s_D\\
            \min \supp (p), & \text{if } s \le s_D
            \end{array}\right\}
    \end{equation*}
\end{lemma}

An immediate corollary from the lemma is that the zero-sum game in Equation \eqref{eq_thrm_1_relaxed_utils} captures \emph{exactly} the utilities of the interaction in case the chosen buyer-optimal segmentation is the equal revenue one, meaning that the desginer cannot lower the regret below $\frac{U^*(0)}{e}$ using equal revenue BBM segmentation.

\begin{corollary}
    Assume that the designer is forced to use BBM segmentations with an equal revenue implementation. Then, the overall regret (subject to this restriction) is $\frac{U^*(0)}{e}$.
\end{corollary}

\begin{proof}
    Going back to the proof of Theorem \ref{main_res_theorem}, the term in Equation \eqref{eq_thrm_1_orig_utils} is the adversary's utility in the resulting zero-sum game and, not just an upper bound (since the designer is restricted to BBM segmentations). Moreover, by Lemma \ref{lemma:equal_revenue_segments}, the "relaxed" utility term in Equation \eqref{eq_thrm_1_relaxed_utils} equals to the adversary's utility, since whenever the designer's guess is "over-optimistic" (namely, $s>s_D$) the buyer surplus indeed becomes zero. Hence, solving for the game defined by  Equation \eqref{eq_thrm_1_relaxed_utils} yields the actual value of the game (and correspondingly, the regret) in the case where the designer must use a BBM segmentation with equal revenue implementation.
\end{proof}

\begin{example}
    \label{example:unif3_cont}
Returning to Example \ref{example:unif3}, the equal revenue choice yields exactly the regret of $\frac{U^*(0)}{e} \approx 0.245$.
Alternatively, in a simulation, using a generic LP solver to find an optimal segmentation for a given $s_D$, the resulting zero-sum game (with real-valued action spaces for the designer and the adversary) admits a value of $\approx 0.229$.
This suggests that the specific implementation of the BBM segmentation for a given $s_D$ affects the actual regret it achieves. Nevertheless, the regret gap induced by the choice of the specific $s_D$-optimal segmentation ($\approx 0.245-0.229$), in this example, is substantially smaller than the gap resulting from restricting the designer to BBM segmentations ($\approx 0.229-0.165$).

\end{example}

\section{Experimental Results}
\label{expriments_section}
We now turn to evaluate our robust segmentation by computing the surplus it guarantees to buyers in different markets. For this experimental setting, we consider the Bayesian model, and given a pair of buyer and seller distributions we compute the \emph{expected optimal buyer surplus} $ \E_{\us \sim F} [U^*(\us)]$, and the \emph{expected robust buyer surplus} $\E_{\us \sim F} [U(\sigma,\us)]$, where $\sigma$ is the robust market segmentation that guarantees the overall regret upper bound of Theorem \ref{main_res_theorem}. While the first reflects the surplus obtained by a market designer who knows the exact valuation of the seller (and is realized by applying the algorithm of \citealp{bergemann2015}), the latter is the result when the designer is devoid of any knowledge, including the seller's valuation distribution. This is the surplus that is obtained by our main result.

We evaluate our robust segmentation with respect to markets in which the seller and buyer distributions are identical, meaning that the seller's valuation $\us$ is sampled from the distribution $\mu$ (independently from the buyer's valuation). Following the work of \cite{coad2009distribution}, we consider two distribution families that represent actual product quality distributions in markets: the Pareto distribution and the Lognormal distribution. In the following simulations, we performed a discretization of these continuous distributions, using $n=15$ discrete values that approximate the continuous distribution.\footnote{Transforming a continuous distribution $F$ into a discrete random variable with support of size $n$ is done by taking the values to be $b_i = \frac{i}{n} \cdot F^{-1}(1-\epsilon)$ for a small $\epsilon > 0$, with weights corresponding to $F(b_i) - F(b_{i-1})$, where we define $b_0 = 0$.} For the Pareto distribution, we run simulations with varying parameter \texttt{alpha}, and for the lognormal distribution, we fix the expectation of the distribution across all experiments and run simulations with varying parameter \texttt{sigma}.\footnote{More precisely, \texttt{sigma} is the standard deviation for the random variable $X \sim 
\mathcal{N}(\texttt{mu}, \texttt{sigma}^2)$, for which $Y=e^X$ has a lognormal distribution (where we fix $\texttt{mu}=1$).} Figure \ref{fig:experiment_fig} shows the expected optimal surplus achieved by the algorithm of \cite{bergemann2015} (assuming $\us$ is known), and the expected robust surplus achieved by our segmentation (that relies on knowing the true $\us$). Notably, for all tested distributions, the robust surplus provides a good approximation of the optimal surplus. 

\begin{figure}[t]
  \centering
  \begin{subfigure}[b]{0.45\linewidth}
    \includegraphics[width=\linewidth]{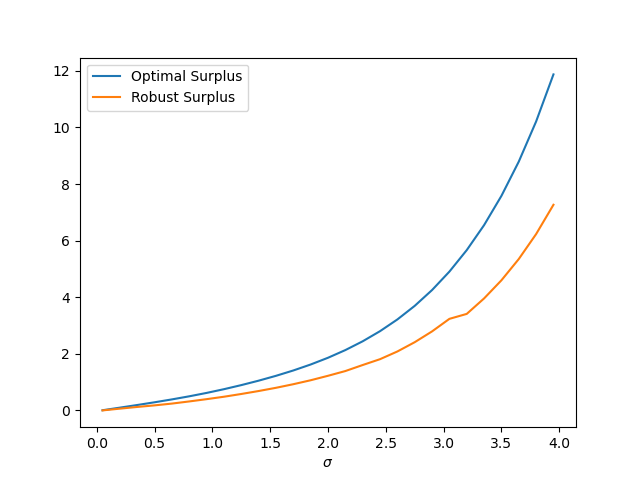}
    \caption{Seller and buyer have a Lognormal distribution with fixed expectation and standard deviation \texttt{sigma}.}
    \label{fig:experiment_lognormal_subfig}
  \end{subfigure}
  \hspace{0.5cm} 
  \begin{subfigure}[b]{0.45\linewidth}
    \includegraphics[width=\linewidth]{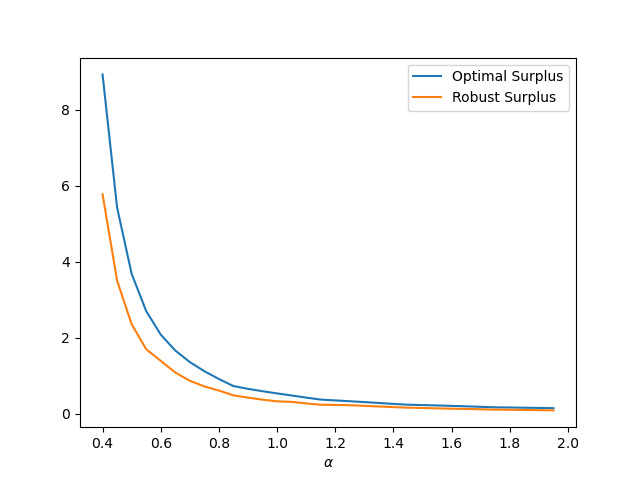}
    \caption{Seller and buyer have a Pareto distribution with parameter \texttt{alpha}.}
    \label{fig:experiment_pareto_subfig}
  \end{subfigure}
  \caption{Expected optimal and robust surplus as a function of a shared seller and buyer distribution parameter.}
  \label{fig:experiment_fig}
\end{figure}

\begin{figure}[t]
  \centering
  \begin{subfigure}[b]{0.45\linewidth}
    \includegraphics[width=\linewidth]{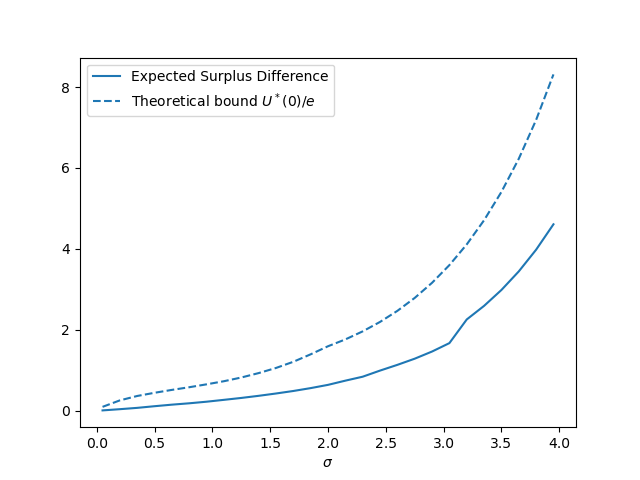}
    \caption{Seller and buyer have a Lognormal distribution with fixed expectation and standard deviation \texttt{sigma}.}
    \label{fig:experiment_lognormal_subfig2}
  \end{subfigure}
  \hspace{0.5cm} 
  \begin{subfigure}[b]{0.45\linewidth}
    \includegraphics[width=\linewidth]{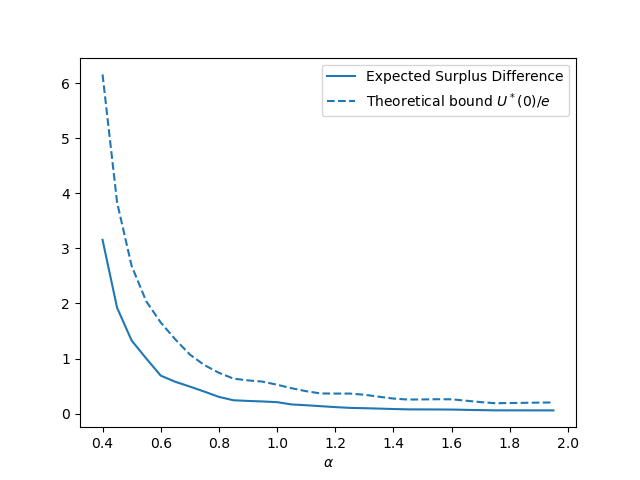}
    \caption{Seller and buyer have a Pareto distribution with parameter \texttt{alpha}.}
    \label{fig:experiment_pareto_subfig2}
  \end{subfigure}
  \caption{Expected difference between optimal and robust surplus, and its theoretical upper bound, as a function of a shared seller and buyer distribution parameter.}
  \label{fig:experiment2_fig}
\end{figure}

Our main result guarantees that the maximal difference between the optimal surplus and the robust surplus is at most $\frac{U^*(0)}{e}$. Our experiment reveals that in practice, for these realistic seller and buyer distributions, the expected difference between the two terms is significantly lower: Figure \ref{fig:experiment2_fig} shows this expected difference $\E_{\us \sim F}[U^*(\us) - U(\sigma,\us)]$ compared to the theoretical upper bound $\frac{U^*(0)}{e}$, and demonstrates that the actual expected difference is much lower in practice. This implies that in some settings (which might be considered as realistic, following \citealp{coad2009distribution}) our robust segmentation performs significantly better than our worst-case bound on the overall regret.

\section{Extensions}
\label{sec:extensions}

Thus far, we have introduced a model of robust price discrimination, where the seller’s valuation is unknown to the market designer, who aims to segment the market to minimize worst-case regret. By formulating the interaction as a zero-sum game, where an adversary selects the true seller valuation and the designer chooses a hypothetical one to guide segmentation, we derived an upper bound on the overall regret. In this section, we present two variants of the model and demonstrate how the same technique can be adapted to obtain corresponding regret upper bounds.

\subsection{The Case of Restricted Seller Valuation}
\label{subsec: restricted_seller_val}

Our construction of the robust market segmentation (and its corresponding regret upper bound) assumes that the designer knows nothing about the seller's product valuation. However, in many practical applications, the designer might have some partial knowledge that could be useful in the context of regret-minimization. One notable case is when the designer knows that the seller's valuation $s$ lies within a given interval $[\underline{s}, \overline{s}] \subset [0, b_{n-1}]$. 
We now show that if the designer knows such lower and upper bounds on the seller's valuation, she can adjust her segmentation to obtain a tighter upper bound on the overall regret.

A first immediate observation is that, without loss of generality, one can normalize the lower bound $\underline{s}$ to be zero, since by subtracting $\underline{s}$ from all of the buyer and seller valuations, we get an equivalent instance of the problem in terms of regret. We therefore focus on the case where the designer only knows that the seller's valuation $s$ is bounded from above by some threshold value $\overline{s} \in (s^*, b_{n-1}]$.\footnote{Note that if $\overline{s} \le s^*$, then the no information policy becomes dominant for the designer, and $\overline{s} > b_{n-1}$ does not provide any useful information for the designer as the optimal strategy of the adversary is supported only on valuations below $b_{n-1}$.}

The case of a known upper bound $\overline{s}$ translates into restricting the strategy space of the adversary in the auxiliary zero-sum game: instead of choosing $s \in [0,b_{n-1}]$, the adversary is now limited to choosing $s \in [0, \overline{s}]$.
This reflects the designer’s partial knowledge about the seller’s valuation and allows for a more refined regret analysis.
Below, we present an adaptation of our main technical lemma to this restricted setting (the proof is deferred to Appendix \ref{app:l2_restricted}):

\begin{lemma}
    \label{main_tech_lemma_restricted_adv}
    Consider the zero-sum game described in Lemma \ref{main_tech_lemma}, except that now Player 1 can only choose $x$ from the interval $[0, \tau]$ for some $\tau \in (\alpha,\beta]$. 
    Let $\delta \in (\alpha,\beta)$ be the unique value such that $u(\delta) = \frac{u(0)}{e}$.
    If $\tau \ge \delta$, then the value of the game is $\frac{u(0)}{e}$ and the optimal strategy $g$ described in Lemma \ref{main_tech_lemma} is still an optimal strategy for Player 2. Otherwise, the value of the game is given by:
    \begin{equation*}
        Val(\tau) = u(\tau) \cdot \ln\bigg( \frac{u(0)}{u(\tau)} \bigg)
    \end{equation*}
    An optimal strategy of Player 2 is given by the following density function:
    \begin{equation*}
        \tilde{g}(y) = \left\{\begin{array}{lr}
        g(y), & \text{for } 
            \alpha \le y \le \tau \\
            0, & \text{otherwise}
        \end{array}\right\}
    \end{equation*}
    with an atom of $1-G(\tau)$ on $\tau$.
\end{lemma}

As expected, the value of the game and the optimal strategies of the players remain unchanged when the restriction on the adversary's strategy space does not exclude any part of the support of the optimal (unrestricted) strategy (this is the case where $\tau \ge \delta$). When $\tau$ decreases below $\delta$, then the value of the game also decreases, aligning with the intuition that imposing stricter constraints on the adversary (the maximizing player in the lemma) reduces their ability to induce regret, thereby lowering the game’s value.
To see that the value is monotone in $\tau$, observe that:
\begin{equation*}
    \frac{\partial Val}{\partial \tau} = u'(\tau) \bigg[ \ln\bigg( \frac{u(0)}{u(\tau)} \bigg) - 1 \bigg] = u'(\tau) \cdot \ln\bigg( \frac{u(0)}{e\cdot u(\tau)} \bigg) 
\end{equation*}
For $\tau < \delta$, we have $u(\tau) > \frac{u(0)}{e}$, implying that $\ln\bigg( \frac{u(0)}{e\cdot u(\tau)} \bigg) < 0$. Combining with $u'(\tau) < 0$, we get that the derivative of the value in $\tau$ is positive, hence lowering the threshold $\tau$ reduces the game value. We are now ready to state our result for the case of known upper bound $\overline{s}$ on the seller's valuation:

\begin{proposition}
\label{main_res_prop_restricted}
Assume that the designer knows that the seller's valuation $s$ is bounded from above by $\overline{s} \in (s^*, b_{n-1}].$
    Then, the overall regret is bounded from above by $\min \bigg\{\frac{U^*(0)}{e}, U^*(\overline{s}) \cdot \ln\big( \frac{U^*(0)}{U^*(\overline{s})} \big) \bigg\}$, and this bound is attained by a BBM market segmentation.
\end{proposition}

\begin{proof}[\textbf{Proof of Proposition \ref{main_res_prop_restricted}.}]
    By the same arguments presented in the proof of Theorem \ref{main_res_theorem}, the overall regret is bounded from above by the value of a zero-sum game in which the designer chooses $s_D$, the adversary chooses $s$, and the adversary's utility is given by Equation \eqref{eq_thrm_1_relaxed_utils}. However, unlike in Theorem \ref{main_res_theorem} can only choose $s$ from $[0,\overline{s}]$. The result then follows by applying Lemma \ref{main_tech_lemma_restricted_adv} with $\tau = \overline{s}$.
\end{proof}

\subsection{The Case of Generalized Regret}
\label{subsec:generalized_regret}

While the standard regret minimization framework aims to minimize the worst-case difference between the buyer's optimal surplus $U^*(s)$ under known seller valuation $s$ and the surplus achieved through a market segmentation strategy $U(s, \sigma)$, 
it is also natural to consider more flexible performance criteria. Specifically, market designers may value trade-offs between optimality and robustness differently depending on the competitive landscape, regulatory constraints, or risk preferences. This motivates the consideration of the \textit{generalized regret objective}, also studied by \cite{bergemann2023managing}. For any $\lambda \in (0, 1)$, the $\lambda-$generalized regret is defined as follows:

\begin{equation*}
    R_\lambda(\sigma) \coloneqq \max_{s} \{ \lambda U^*(s) - (1 - \lambda)U(s, \sigma) \}
\end{equation*}

The overall generalized regret is also defined analogously. The tradeoff parameter $\lambda \in (0, 1)$ encodes the designer's tradeoff between maximizing the worst‐case absolute surplus (low $\lambda$) and avoiding high-benchmark scenarios (high $\lambda$). Notice that the generalized regret boils down to the standard regret for $\lambda = \frac{1}{2}$ (up to multiplying the regret by a constant of $2$).

We note that the proof technique used to bound the overall standard regret from above carries over to the case of generalized regret as well. First, we introduce the following adaptation of the technical lemma to the generalized case (the proof is deferred to Appendix \ref{app:l2_general}):

\begin{lemma}
    \label{main_tech_lemma_general_reg}
    Consider the zero-sum game described in Lemma \ref{main_tech_lemma} with the following utility function:
        \begin{equation*}
        v(x,y) = \left\{\begin{array}{lr}
            \lambda u(x), & \text{for } x > y \\
            \lambda u(x)-(1-\lambda)u(y), & \text{for } x \le y
            \end{array}\right\}
        \end{equation*}

    for some $\lambda \in (0,1)$. Then the value of the game is $\lambda u(0) \cdot e^{- \frac{1-\lambda}{\lambda}}$, and it can be guaranteed to Player 2 by playing a mixed strategy with the following density function:
    \begin{equation*}
        g(y) = \left\{\begin{array}{lr}
        -\frac{\lambda}{1-\lambda} \cdot\frac{u'(y)}{u(y)}, & \text{for } 
            \alpha \le y \le \delta \\
            0, & \text{otherwise}
        \end{array}\right\}
    \end{equation*}
    for $\delta \in (\alpha,\beta)$ such that $u(\delta) = u(0) \cdot e^{- \frac{1-\lambda}{\lambda}}$.
\end{lemma}

Notice that setting $\lambda = \frac{1}{2}$ indeed recovers the result of Lemma \ref{main_tech_lemma}, up to the missing multiplicative factor of $2$ in the game value. We can now use Lemma \ref{main_tech_lemma_general_reg} to derive an upper bound on the overall $\lambda-$generalized regret:

\begin{proposition}
\label{main_res_prop_general}
    The overall $\lambda-$generalized regret is bounded from above by $\lambda U^*(0) \cdot e^{- \frac{1-\lambda}{\lambda}}$, and this bound is attained by a BBM market segmentation.\footnote{To get the same scale as in Theorem \ref{main_res_theorem}, one should multiply the overall regret by $2$.}
\end{proposition}

\begin{proof}[\textbf{Proof of Proposition \ref{main_res_prop_general}.}]

Consider the same zero-sum game between the designer and the adversary (as in the proof of Theorem \ref{main_res_theorem}), except now the adversary's utility is given by:

\begin{equation*}
    v(\us,s_D) = \lambda U^*(\us) - (1-\lambda) U(\sigma^*(s_D),\us)
\end{equation*}

The value of this game is an upper bound on the overall $\lambda-$generalized regret, by similar arguments as in Theorem \ref{main_res_theorem}. By similar argument as in the proof of Theorem \ref{main_res_theorem}, we bound the utility from above by the following term:

\begin{equation*}
    \tilde{v}(\us,s_D) = \left\{\begin{array}{lr}
        \lambda U^*(\us), & \text{for } \us > s_D \\
        \lambda U^*(\us)-(1-\lambda) U^*(s_D), & \text{for } \us \le s_D
        \end{array}\right\}
\end{equation*}

The value of this game is, by the exact same argument, an upper bound on the overall $\lambda-$generalized regret. Using the same arguments as in the proof of Theorem \ref{main_res_theorem}, we can eliminate strategies $s, s_D > b_{n-1}$, which leave us with a game that falls into the conditions specified in Lemma \ref{main_tech_lemma_general_reg}. Applying the lemma then concludes the proof.

\end{proof}

As $\lambda$ increases, the designer’s objective puts ever more weight on the benchmark term $U^*(s)$, the surplus she could only achieve if she actually knew $s$. However, since she cannot influence $U^*(s)$ through her choice of $\sigma$, a larger $\lambda$ effectively magnifies a source of “regret” that lies entirely outside her control. In the extreme, as $\lambda \to 1$, the worst‐case generalized regret approaches $\max_s U^*(s) = U^*(0)$, which is the maximum possible benchmark value no matter how clever the segmentation is. Thus it is inevitable (and intuitive) that the minimal achievable $R_\lambda$ rises sharply with $\lambda$: by demanding closer alignment with an unattainable ideal, the designer is forced to bear larger residual loss in every scenario. 
While Proposition \ref{main_res_prop_general} only provides an upper bound, it is reasonable to view a similar behavior as we could expect from the actual regret.
Figure \ref{fig:Rlambda} illustrates the behavior of the generalized regret upper bound as a function of the tradeoff parameter $\lambda$ (for the case where $U^*(0)$ is normalized to $1$).

\begin{figure}[t]
  \centering
  \begin{tikzpicture}
    \begin{axis}[
      width=0.5\textwidth,
      xlabel={\(\lambda\)},
      ylabel={Upper bound on \(R_\lambda\)},
      xmin=0, xmax=1,
      ymin=0, ymax=1,
      samples=200,
      domain=0.01:0.99,
      xtick=\empty,
      ytick=\empty,
      axis lines = left,
    ]
      \addplot [thick, orange] {x * exp(-(1 - x) / x)};
    \end{axis}
  \end{tikzpicture}
  \caption{The upper bound on the overall $\lambda-$generalized regret as a function of $\lambda$, for \(U^*(0)=1\).}
  \label{fig:Rlambda}
\end{figure}

\section{Discussion}
\label{conclusions_section}
This work studies the celebrated price discrimination problem of \cite{bergemann2015} under the relaxation of a major assumption of complete information about the valuation of the seller. 
We adopt the robust perspective, according to which the designer aims to achieve a segmentation that guarantees a low regret, i.e., a segmentation that performs well (compared to the optimal one) regardless of the actual seller valuation.
Our main results suggest that our two-stage approach of sampling a seller valuation and acting \emph{as if} this was the true seller valuation, obtains an upper bound on the overall regret. We further show that this bound is indeed tight in the binary buyer type case. Lastly, we demonstrated that in some realistic markets, our approach yields a regret that is even better than the worst-case guarantee. 

We argue that many realistic applications of third-degree price discrimination have the property of partial (or even completely no) information of the seller's valuation. As an example, consider an online retail platform, such as Amazon or eBay, that can control the information available to the seller about potential buyers by selective presentation of users' information. Such platforms may be interested in preserving the average satisfaction of their uses, which translates into maximizing the buyers' surplus in our price discrimination model.
In this scenario, it is unreasonable to assume that the market designer knows exactly how much the seller appreciates her product. Our approach enables the platform to achieve a great level of user satisfaction regardless of the actual seller type, which may be considered an extremely strong guarantee.

Interestingly, our incomplete information model of price discrimination can also be reinterpreted as a framework for dealing with market shocks. In this alternative view, the seller's outside option is fixed and normalized to zero, while the designer only knows the buyer valuations up to an additive aggregated noise level, $\us$. This noise, which encapsulates unforeseen market fluctuations, remains unknown until after the platform has committed to its signaling policy and is subsequently revealed to all agents. In this setting, the designer’s objective is to craft a market segmentation that performs robustly across all possible realizations of $\us$. Such a model is particularly relevant in scenarios where external factors—such as sudden shifts in consumer sentiment, seasonal variations, or broader economic disturbances—induce aggregate changes in buyer valuations. By incorporating the potential impact of these market shocks into the design process, our approach ensures that the market segmentation remains effective even under significant uncertainty, thereby broadening its applicability to a variety of real-world market environments.

A natural question is whether the designer can learn something about the seller’s valuation ex post, after observing the interaction outcome. A key observation is that if the designer applies a BBM segmentation based on a hypothesized seller value \(s_D\), then she can infer whether the true valuation \(s\) lies above or below \(s_D\), based on the price the seller chooses (at least when the realized posterior induces price options both below and above the monopolistic price).\footnote{Following \cite{bergemann2015}, in any optimal segmentation for a given $s_D$, the seller is indifferent between the lowest price in the segment and the monopolistic price $\pi(\mu; s)$. Thus, if the designer applies a BBM segmentation for $s_D$ and observes both the signal realization and the seller's chosen price, she can infer whether $s \le s_D$ (seller chooses the lowest price) or $s > s_D$ (seller prefers a higher price).} 
This observation can be useful, for instance, in repeated interactions with a myopic seller (whose valuation is unknown at the outset), who sets price greedily based on the current round's realized segment. By iteratively applying BBM segmentations and updating her knowledge based on outcomes, the designer can exponentially narrow in on \(s\) while still achieving bounded regret in each round (leveraging the results of Subsection \ref{subsec: restricted_seller_val}).

\paragraph{Future directions} 
Providing a lower bound on the overall regret for an arbitrary number of buyer types is left as an interesting future direction.
Another direction is the study of multiplicative regret (or "competitive ratio") in the price discrimination model. Notably, in the case of multiplicative regret, the utilities in the resulting zero-sum game will contain terms of the form $u(x) / u(y)$, which means that our technical lemma cannot be easily adjusted. The reason is that a core property of the resulting game in the case of additive regret is the additivity of $u(x)$ and $u(y)$, which enables the players to cancel out each other dependencies on the behavior of $u$ in the entire interval of the strategy space.

The complexity of computing the regret-minimizing value, and relatedly the regret-minimizing policy, remains an interesting open problem.
Regarding the problem of computing an exact regret-minimizing policy: Since we cannot bound the number of signals required in a regret-minimizing policy (it can be infinite even in the case of two valuations), it remains unclear whether such policies admit an alternative, succinct representation. We do not see a reason to believe it does in the general case. Thus, we do not see a reason that the problem will belong to NP.
This leads us to consider approximations. The problem of computing an $\epsilon$-regret minimizing policy can be solved via an LP after discretizing the signaling space. However, the size of the LP in this (straightforward) procedure grows exponentially with the number of different valuations. It remains an interesting open problem whether more clever algorithms can compute approximate regret-minimizing policies in time that depends polynomially on the number of valuations.
In a different setting of forecast aggregation, in which the regret minimization problem boils down to a zero-sum game with infinite-dimensional action spaces, as it is in our setting, the computational aspects of the problem have been considered in \cite{guo2024algorithmic}. We believe that the techniques of \cite{guo2024algorithmic} might be beneficial to deduce computational insight, also in our case.

Finally, it is natural to consider an alternative setting in which the designer offers a menu of segmentations to the seller. \cite{SYnew} adopt this approach in the case of a seller who holds private partial information about the buyer's valuation. They show that this private information cannot be screened. Whether the seller's value for the good can be screened in our model remains an interesting follow-up problem.

\section*{Acknowledgements}
Itai Arieli gratefully acknowledges support from the Israel Science Foundation (grant agreements \#2029464 and \#2071717).
Yakov Babichenko gratefully acknowledges support from the Bi-national Science Foundation (NSF-BSF grant \#2021680) and the Israel Science Foundation (grant agreement \#2061/24).
The work by Moshe Tennenholtz and Omer Madmon was supported by funding from the European Research Council (ERC) under the European Union’s Horizon 2020 research and innovation programme (grant agreement \#740435).
We would like to thank the associate editor and the anonymous reviewers for their helpful comments. An abstract of this paper appeared in the proceedings of Theoretical Aspects of Rationality and Knowledge (TARK) 2025.

\section*{Declaration of Generative AI Usage in the Writing Process}
During the preparation of this work, the authors used ChatGPT for proofreading. After using this tool, the authors reviewed and edited the content as needed. The authors take full responsibility for the content of the publication.

\section*{Declaration of Competing Interest}
The authors declare that they have no known competing financial interests or personal relationships that could have appeared to influence the work reported in this paper.

\section*{Data Availability}
No data was used for the research described in the article.

\bibliographystyle{plainnat}
\bibliography{references}

\appendix

\section{Omitted Proofs}

\subsection{Proof of Lemma \ref{key_prop_lemma}}\label{app:l1}

\begin{proof}
    
First, notice that $U^*(\cdot)$ satisfies absolute continuity (as a maximum of linear functions on a closed interval), and non-negativity (by definition). As for its differentiability and monotonicity properties, we first recall that the optimal market segmentation $\sigma^*(\us)$ takes the following form: if $b_1 = \pi(\mu;\us)$ then the optimal market segmentation consists of the prior buyer distribution solely, i.e. $\sigma^*(\mu) = 1$. Otherwise, it consists of a set of at most $n$ posteriors, such that at any posterior $p$ the seller is indifferent between the monopolistic price $\pi(\mu;\us)$ and the lowest buyer type in the support of $p$. Since the optimal price $\pi(\mu;\us)$ is weakly increasing in $\us$, we get that there must exist some $s^*$ such that no segmentation is optimal if and only if $\us < s^*$. 

Next, note that $U^*(\cdot)$ is differentiable up to a finite number of points, corresponding to the set of points for which the seller is indifferent between several prices.
It can be now seen from Equation \eqref{eq: optimal surplus} that whenever no segmentation is optimal (namely $\us < s^*$, which means $i^* = 1$), $U^*(\cdot)$ is constant and equals $\sum_{j=1}^n \mu_j \cdot (b_j - b_1)$. 

Next, assume that $s^* \le \us < b_{n-1}$, and in particular $i^*>1$. It is clear that in this case $i^* < n$ (since setting the price $b_n$ yields zero utility for the seller, while setting it to e.g. $b_{n-1}$ yields some positive utility, since buyers of type $b_n$ will buy the product). In this case, the optimal buyer surplus can be written as follows:

\begin{equation*}
    U^*(\us) = \sum_{j<i^*} \mu_j \cdot \max \{ b_j - \us, 0 \} + \sum_{j \ge i^*} \mu_j \cdot \bigr( \max \{ b_j - \us, 0 \} - (b_{i^*} - \us) \bigl)
\end{equation*}

For all $j \ge i^*$ it holds that $b_j \ge b_{i^*} \ge \us$, and therefore $\max \{ b_j - \us, 0 \} - (b_{i^*} - \us) = b_j - b_{i^*}$. Hence, the rightmost sum is independent of $\us$. In addition, there exists at least one $j < i^*$ for which $b_j > \us$ 
and the leftmost sum is strictly decreasing as a sum of strictly decreasing functions.
Finally, it is clear that when $\us \ge b_{n-1}$, the optimal price corresponds to the highest buyer type (namely, $i^* = n$), and $U^*(\us) = 0$.

\end{proof}

\subsection{Proof of Lemma \ref{main_tech_lemma}}\label{app:l2}

\begin{proof}

Let us consider a mixed strategy profile, in which player 1 chooses a distribution with density $f$ and CDF $F$, and player 2 chooses a distribution with density $g$, both with support $[\alpha,\delta]$ for some $\alpha < \delta < \beta$ (where the distribution $f$ also has an atom at zero). Then, $(f,g)$ is a mixed Nash equilibrium if the following indifference conditions hold:

\begin{enumerate}
    \item $v(x,g)$ is independent of $x$.
    \item $v(f,y)$ is independent of $y$, for $y\in[\alpha,\delta]$.
    \item $v(f,y) \ge v(f,\delta)$ for $y\notin[\alpha,\delta]$.
\end{enumerate}

To satisfy the first condition, we require:

\begin{equation*}
    \frac{\partial v(x,g)}{\partial x} = 0
\end{equation*}

Since $v(x,g) = u(x) - \int_{y=x}^\delta g(y) u(y) dy$, the above condition holds if and only if:

\begin{equation*}
    u'(x) + g(x) u(x) = 0 \Leftrightarrow g(x) = - \frac{u'(x)}{u(x)}
\end{equation*}

Notice that indeed $g \ge 0$, since $u$ is non-increasing and nonnegative. Now, $\delta$ can be found using the normalization constraint of the distribution $g$:

\begin{equation*}
    1 = \int_{x=\alpha}^\delta g(x) dx = - \int_{x=\alpha}^\delta \frac{u'(x)}{u(x)} \cdot dx \underset{(1)}{{=}} \ln(u(\alpha)) - \ln(u(\delta))
\end{equation*}
\begin{equation*}
    \Rightarrow \ln(u(\delta)) = \ln(u(\alpha)) - \ln(e) = \ln\Bigr(\frac{u(\alpha)}{e}\Bigr) \Rightarrow u(\delta) = \frac{u(\alpha)}{e} \underset{(2)}{{=}} \frac{u(0)}{e}
\end{equation*}

where $(1)$ follows from the Newton-Leibniz formula (and the absolute continuity of $u$), and $(2)$ follows from the fact that $u$ is constant in range $[0,\alpha]$. Notice that $\delta$ is indeed contained in the open interval $(\alpha,\beta)$.\footnote{To see why, note that
$\ln(u(\alpha)) - \ln(u(\delta)) \to \infty$ as $\delta \to \beta$, and
$\ln(u(\alpha)) - \ln(u(\delta)) \to 0$ as $\delta \to \alpha$.
Hence, from continuity, there must exist some $\delta \in (\alpha, \beta)$ for which $\ln(u(\alpha)) - \ln(u(\delta))=1$.
}

As for the second condition, note that for any $y\in[\alpha,\delta]$:

\begin{equation*}
    v(f,y) = \E_{x \sim f} [u(x)] - F(y) u(y)
\end{equation*}

The condition holds for $F(y) = \frac{c}{u(y)}$ for some constant $c > 0$ (note that the distribution has an atom at zero). Note that $F$ is a valid CDF since $u$ is strictly decreasing in $[\alpha,\delta]$. Now it is left to find $c$ for which this indifference holds for any $y\in[\alpha,\delta]$:

\begin{equation*}
    F(\delta) = 1 \Leftrightarrow c = u(\delta)
\end{equation*}

If $y < \alpha$, it holds that $F(y) = 0$ and clearly $v(f,y)$ increases - hence it is not beneficial for player 2 who aims to minimize $v$. Lastly, note that for $y > \delta$, the term $F(y) u(y)$ decreases, hence $v(f,y)$ increases - and therefore player 2 does not assign a positive probability for any $y\notin[\alpha,\delta]$ when player 1 plays $f$.

Overall, we obtain that $(f,g)$ is a mixed Nash equilibrium, and the value of the game is $u(\delta) = \frac{u(0)}{e}$, as $\delta$ is the highest action played by player 2 with positive probability, and player 1 is indifferent and might as well play $\delta$ with probability 1. In that case, $x>y$ with probability $1$, and hence the value is $u(\delta)$.

\end{proof}

\subsection{Proof of Lemma \ref{bbm_lemma}}

\begin{proof}
\label{app:3}

First, it is clear that if $\us > b_1$ the seller never sells the product regardless of the segmentation, and therefore the buyer surplus is always zero. Assume now that this is not the case. Consider the equivalent market segmentation problem without seller valuation, and with buyer types $\tilde{b}_i \coloneqq b_i - \us$ for $i \in \set{1,2}$. From \cite{bergemann2015}, if $\tilde{b}_1$ is an optimal price in $(\mu, 1-\mu)$, then no segmentation is optimal, and the buyer surplus is simply $1-\mu$. This happens if and only if:

\begin{equation*}
    \tilde{b}_1 \ge \tilde{b}_2 (1-p) \Leftrightarrow \us \le b_2 - \frac{1}{\mu}
\end{equation*}

Otherwise, a segmentation that maximizes the buyer surplus is of the following form:

\begin{equation*}
\begin{split}
    (0,1) \text{ w.p. } \alpha, \\
    (p,1-p) \text{ w.p. } 1-\alpha 
\end{split}
\end{equation*}

such that at the posterior $(p,1-p)$ the seller is indifferent between prices $\tilde{b}_1$ and $\tilde{b}_2$, and Bayes plausibility holds. The buyer surplus is $(1-p)(1-\alpha)$. The seller's indifference condition yields:

\begin{equation}
\label{bbm_indf}
    \tilde{b}_1 = \tilde{b}_2 (1-p) \Rightarrow 1-p = \frac{\tilde{b}_1}{\tilde{b}_2}
\end{equation}

and from the Bayes plausibility condition:
\begin{equation}
\label{bbm_bp}
    x(1-\alpha) = \mu \Rightarrow 1-\alpha = \frac{\mu}{x}
\end{equation}

Combining Equation \eqref{bbm_indf} and Equation \eqref{bbm_bp}, we get:

\begin{equation*}
    (1-p)(1-\alpha) = \frac{\tilde{b}_1}{\tilde{b}_2} \cdot \frac{\mu}{1 - \frac{\tilde{b}_1}{\tilde{b}_2}} = \tilde{b}_1 \mu = (b_1 - \us) \mu
\end{equation*}

Finally, notice that for $\us = b_2 - \frac{1}{\mu}$, $(b_1 - \us) \mu = 1 - \mu$.
    
\end{proof}

\subsection{Proof of Lemma \ref{lemma:equal_revenue_segments}}
\label{app:proof_of_lemma_eq_rev}

\begin{proof}
Assume $p \in \Delta(B)$ is an equal revenue segment with respect to $s_D$, i.e., for all $i, j \in [n]$, it holds that:
\[
(b_i - s_D) \sum_{k=i}^{n} p_k = (b_j - s_D) \sum_{k=j}^{n} p_k.
\]

Let $F_i = \sum_{k=i}^{n} p_k$ for each $i \in [n]$.

\textbf{Case I:} $s = s_D + \delta$ for some $\delta > 0$.  
Let $j$ be the index of $\max \supp(p)$.  
We compare the revenue from two prices $b_i$ and $b_j$:
\begin{align*}
(b_i - s) F_i - (b_j - s) F_j 
&= (b_i - s_D - \delta) F_i - (b_j - s_D - \delta) F_j \\
&= (b_i - s_D) F_i - \delta F_i - (b_j - s_D) F_j + \delta F_j \\
&= \left[ (b_i - s_D) F_i - (b_j - s_D) F_j \right] + \delta (F_j - F_i).
\end{align*}

By the equal revenue property, the first term is zero. Since $j$ is the maximal index in the support, $F_j < F_i$ for $i < j$, so $F_j - F_i < 0$, and hence the expression is negative.  
Thus, the optimal price is $b_j = \max \supp(p)$.

\textbf{Case II:} $s = s_D - \delta$ for some $\delta > 0$ (it is clear that if $s=s_D$, then the optimal price remains unchanged).  
Let $j$ be the index of $\min \supp(p)$.  
Similarly, for any $i > j$:
\begin{align*}
(b_i - s) F_i - (b_j - s) F_j 
&= (b_i - s_D + \delta) F_i - (b_j - s_D + \delta) F_j \\
&= (b_i - s_D) F_i + \delta F_i - (b_j - s_D) F_j - \delta F_j \\
&= \left[ (b_i - s_D) F_i - (b_j - s_D) F_j \right] + \delta (F_i - F_j).
\end{align*}

Again, the first term is zero by the equal revenue condition. Since $j$ is the minimal index in the support, $F_j > F_i$, so $F_i - F_j < 0$, and the expression is negative.  
Hence, the optimal price is $b_j = \min \supp(p)$.

\end{proof}

\subsection{Proof of Lemma \ref{main_tech_lemma_restricted_adv}}\label{app:l2_restricted}

\begin{proof}
    The proof relies on notations and arguments analogous to those presented in the proof of Lemma \ref{main_tech_lemma} (Appendix \ref{app:l2}). 
    We will refer to the proof of Lemma \ref{main_tech_lemma} as the "original proof", as many of its arguments will also be applied here.
    First, notice that if $\tau \ge \delta$, both player can still play their optimal strategies presented in Lemma \ref{main_tech_lemma}, therefore the value of the game remains $\frac{u(0)}{e}$. It is therefore left to consider the case where $\alpha < \tau < \delta$, an optimal strategy of player 2 that is supported on $[\alpha,\tau]$ must have a density $-u'(x)/u(x)$ in order for the indifference condition of player 1 to hold. By the definition of the utilities in the game, it is also clear that placing the remaining mass $1-G(\tau)$ on $\tau$ preserves this indifference.
    As for the optimal strategy of player 1, notice that by similar arguments as in the original proof, the optimal strategy must be of the form $F(y) = \frac{c}{u(y)}$, with the constant $c$ that is determined by the normalization constraint now varies as both optimal strategies are now supported on $[\alpha,\tau]$.

    As for the value of the game, notice that by similar arguments as in the original proof, it equals $v(\tau,\tilde{g})$. (where $\tilde{g}$ is the optimal strategy of player 2). A simple calculation yields:

    \begin{equation*}
        Val(\tau) = \int_{y=\alpha}^\tau v(\tau,y) \cdot g(y) \cdot dy + \big(1-G(\tau)\big) \cdot v(\tau,\tau) = u(\tau) \cdot G(\tau)
    \end{equation*}
    \begin{equation*}
        \Rightarrow Val(\tau) = u(\tau) \bigg[ \ln\big(u(0)\big) - \ln\big(u(\tau)\big) \bigg] = u(\tau) \cdot \ln\bigg( \frac{u(0)}{u(\tau)} \bigg)
    \end{equation*}
    
\end{proof}

\subsection{Proof of Lemma \ref{main_tech_lemma_general_reg}}\label{app:l2_general}

\begin{proof}
    The proof relies on notations and arguments analogous to those presented in the proof of Lemma \ref{main_tech_lemma} (Appendix \ref{app:l2}). 
    We will refer to the proof of Lemma \ref{main_tech_lemma} as the "original proof", as many of its arguments will also be applied here.
    For the indifferent condition of player 1, we get:

    \begin{equation*}
    \lambda u'(x) + (1-\lambda) g(x) u(x) = 0 \Leftrightarrow g(x) = - \frac{\lambda u'(x)}{(1-\lambda) u(x)}
\end{equation*}

    By similar arguments as in the original proof, $ g$ is still non-negative. Finding $\delta$ (which now depends on $\lambda)$ using the normalization constraint yields:

    \begin{equation*}
    1 = \int_{x=\alpha}^\delta g(x) dx = - \frac{\lambda}{1-\lambda} \int_{x=\alpha}^\delta \frac{u'(x)}{u(x)} \cdot dx = \frac{\lambda}{1-\lambda} \bigg( \ln(u(\alpha)) - \ln(u(\delta)) \bigg)
    \end{equation*}
    \begin{equation*}
        \Rightarrow \ln(u(\delta)) = \ln(u(\alpha)) - \ln\big(e^\frac{1-\lambda}{\lambda}\big) \Rightarrow u(\delta) = u(0) \cdot e^{- \frac{1-\lambda}{\lambda}}
    \end{equation*}

    For the indifference conditions of player 2, we still get (by the same arguments as in the original proof) that $F(y) = \frac{c}{u(y)}$ for some constant $c > 0$ (which also depends on $\lambda$), again with an atom on zero.
    Applying the exact same arguments as in the original proof yields that $c = u(\delta)$, and that player 2 has no beneficial deviation to strategies outside the support of $f$.
    We therefore obtain that $(f,g)$ is a mixed Nash equilibrium, and the value of the game is $\lambda u(\delta)$, by the same argument as in the original proof.
    
    \end{proof}

\end{document}